\documentclass[11pt]{article}

\usepackage{graphicx}
\usepackage[utf8]{inputenc}
\usepackage{amsthm}
\usepackage[fleqn]{amsmath}
\usepackage{amssymb}
\usepackage[lined,boxed]{algorithm2e}
\usepackage{geometry}                
\allowdisplaybreaks

\DeclareMathOperator*{\argmin}{arg\,min}

\def\R{\mathbb{R}}

\def\E{\mathbb{E}}

\def\sse{\subseteq}

\newtheorem{theorem}{Theorem}

\newtheorem{lemma}{Lemma}

\newcommand\ignore[1]{}

\title{Online Covering with Sum of $\ell_q$-Norm Objectives}

\author{Viswanath Nagarajan\thanks{Industrial and Operations Engineering Department, University of Michigan.} \and Xiangkun Shen$^*$}

\begin{document}
\maketitle

\begin{abstract}
We consider fractional online covering problems with $\ell_q$-norm objectives. The problem of interest is of the form $\min\{ f(x) \,:\, Ax\ge 1, x\ge 0\}$ where $f(x)=\sum_{e} c_e \|x(S_e)\|_{q_e} $ is the weighted sum of $\ell_q$-norms and $A$ is a non-negative matrix. The rows of $A$ (i.e. covering constraints) arrive online over time. We provide an online $O(\log d+\log \rho)$-competitive algorithm where $\rho = \frac{\max a_{ij}}{\min a_{ij}}$ and $d$ is the maximum of the row sparsity of $A$ and $\max |S_e|$. This is based on the online primal-dual framework where we use the dual of the above convex program. Our result expands the class of convex objectives that admit good online algorithms: prior results required a monotonicity condition on the objective $f$ which is not satisfied here. This result is nearly tight even for the linear special case. As direct applications we obtain (i) improved  online algorithms for non-uniform  buy-at-bulk network design  and (ii) the first online algorithm for throughput maximization under $\ell_p$-norm edge capacities. 

\end{abstract}

\section{Introduction}
The online primal-dual approach is a widely used approach for online problems. This involves solving a discrete optimization problem online as follows (i) formulate a linear programming relaxation and obtain a primal-dual online algorithm for it; (ii) obtain an online rounding algorithm for the resulting fractional solution. While this is similar to a linear programming (LP) based approach for offline optimization problems, a key difference is that solving the LP relaxation in the online setting is highly non-trivial. (Recall that there are general polynomial time algorithms for solving LPs offline.) So there has been a lot of effort in obtaining good online algorithms for various classes of LPs: see \cite{AAABN06,BN09,GN14} for pure covering LPs, \cite{BN09} for pure packing LPs and \cite{ABFP13} for certain mixed packing/covering LPs. Such online LP solvers have been useful in obtaining online algorithms for various problems, eg. set cover~\cite{AAABN03}, facility location~\cite{AAABN06}, machine scheduling~\cite{ABFP13}, caching~\cite{BBN12} and buy-at-bulk network design~\cite{ECKP15}.

Recently, \cite{ABCCCG0KNNP16} initiated a systematic study of online fractional covering and packing with {\em convex} objectives; see  also the full versions~\cite{ACP14,BCGNN14,CHK15}. These papers obtained good online algorithms for a large class of fractional convex covering problems. They also  demonstrated the utility of this approach via many applications that could not be solved using just online LPs. However these  results were limited to convex objectives $f: \R^n_+\rightarrow \R_+$ satisfying a monotone gradient property, i.e. $\nabla f(z) \ge \nabla f(y)$ pointwise  for all $z,y\in \R^n$ with $z\ge y$. There are however many natural convex functions that do not satisfy such a gradient monotonicity condition.  Note that this condition requires the Hessian $\nabla^2 f(x)$ to be pointwise non-negative in addition to convexity which only requires $\nabla^2 f(x)$ to be positive semidefinite.

In this paper, we focus on convex functions $f$ that are  sums of different $\ell_q$-norms. This is a canonical class of convex functions with non-monotone gradients and prior results are not applicable; see Section~\ref{sec:ORT} for a more detailed comparison. We show that sum of $\ell_q$-norm functions admit a logarithmic competitive online algorithm. This result is nearly tight because there is a logarithmic lower bound even for online covering LPs (which corresponds to an $\ell_1$ norm objective). We also provide two applications of our result  (i) improved competitive ratios (by two logarithmic factors) for some online non-uniform buy-at-bulk problems studied in~\cite{ECKP15}, and (ii) the first online algorithm for throughput maximization  with $\ell_p$-norm edge capacities (the competitive ratio is logarithmic which is known to be  best possible even in the special case of individual edge capacities). 

Given that we achieve log-competitive online algorithms for sums of $\ell_q$-norms, a natural question is whether such a result holds for all norms. Recall that any norm is a convex function. It turns out that a log-competitive algorithm is not possible for general norms. This follows from a result in \cite{ACP14} which shows an $\Omega (q\log d)$ lower bound for minimizing the objective $\|Bx\|_q$ under covering constraints (where $B$ is a non-negative matrix). It is still  an interesting open question to identify the correct competitive ratio for general norm functions. 

\subsection{Our Results and Techniques}\label{sec:ORT}

We consider the online covering problem 
\begin{equation}\label{eq:covering-prob}
\min \, \left\{ \sum_{e=1}^r c_e\|x(S_e)\|_{q_e}\,:\,  Ax\ge \mathbf{1},\, {x\in \R_+^n}\right\},
\end{equation}
where each $S_e\sse [n]:=\{1,2,\cdots n\}$, $q_e\ge 1$, $c_e\ge 0$ and $A$ is a non-negative $m\times n$ matrix.  For any $S\sse [n]$ and $q \ge 1$ we use the standard notation 
$\|x(S)\|_q = \left( \sum_{i\in S} x_i^q\right)^{1/q}$. We also consider the dual of this convex program, which is the following packing problem:
\begin{equation}\label{eq:packing-prob}
\max\, \left\{  \sum_{k=1}^m y_k\, :\,  A^Ty=\mu,\, \sum_{e=1}^r \mu_e = \mu,\,  \|\mu_e(S_e)\|_{p_e}\le c_e \, \forall e\in [r],\,  y\ge 0\right\}. 
\end{equation}
The values $p_e$ above 
satisfy $\frac{1}{p_e}+\frac1{q_e}=1$; so  $\|\cdot\|_{p_e}$ is the dual norm of $\|\cdot\|_{q_e}$.  This dual can be derived from~\eqref{eq:covering-prob} using Lagrangian duality (see Section~\ref{sec:prelim}).

Our framework captures the classic setting of packing/covering LPs when $r=n$ and for each $e\in [n]$ we have $S_e=\{e\}$ and $q_e=1$. Our main result is:
\begin{theorem}\label{thm:main}
There is an $O(\log d + \log \rho)$-competitive online algorithm for \eqref{eq:covering-prob} and \eqref{eq:packing-prob} where the covering constraints in \eqref{eq:covering-prob}  and variables $y$ in \eqref{eq:packing-prob}  arrive over time. Here $d$ is the maximum of  the row-sparsity of $A$ and $\max_{e=1}^r |S_e|$ and $\rho = \frac{\max\{a_{ij}\}}{\min \{a_{ij}\}}$.
\end{theorem}
 We note that this bound is also the best possible, even in the linear case~\cite{BN09}. For just the covering problem, a better $O(\log d)$ bound is known in the linear case~\cite{GN14} as well as for  convex functions with monotone gradients~\cite{ABCCCG0KNNP16}. 
 
The algorithm in Theorem~\ref{thm:main} is the natural extension of the primal-dual approach for online LPs~\cite{BN09}. We use the gradient $\nabla f(x)$ at the current primal solution $x$ as the cost function, and use this to define a multiplicative update for the primal. Simultaneously, the dual solution $y$ is increased additively. This algorithm is in fact identical to the one in~\cite{ABCCCG0KNNP16} for convex functions with monotone gradients. See Algorithm~\ref{alg:primal-dual} for the formal description. The contribution of this paper is in the analysis of this algorithm, which requires new ideas to deal with non-monotone gradients. 
 
\paragraph*{Limitations of  previous approaches~\cite{ABCCCG0KNNP16}}  Recall that the general convex covering problem is 
$$\min \, \left\{ f(x) \,:\,  Ax\ge \mathbf{1},\, {x\in \R_+^n}\right\},$$
where $f:\R^n_+\rightarrow \R_+$ is a convex function. Its dual is: 
$$\max\, \left\{  \sum_{k=1}^m y_k - f^*(\mu) \, :\,  A^Ty=\mu,\,  y\ge 0\right\}, $$
where $f^*(\mu) = \max_{x\in \R^n_+} \{ \mu^Tx - f(x)\}$ is the Fenchel conjugate of $f$.  When $f$ is the sum of $\ell_q$-norms, these primal-dual convex programs reduce to \eqref{eq:covering-prob} and \eqref{eq:packing-prob}. 

We restrict the  discussion of prior techniques to functions $f$ with $\max_{x\in \R^n_+} \frac{x^T \nabla f(x)}{f(x)}\le 1$ because this condition is satisfied by sums of $\ell_q$ norms.\footnote{The result in \cite{ABCCCG0KNNP16} also applies to other convex functions with monotone gradients, but the competitive ratio depends exponentially on $\max_{x\in \R^n_+} \frac{x^T \nabla f(x)}{f(x)}$.}  
At a high level, the analysis in~\cite{ABCCCG0KNNP16} uses the gradient monotonicity to prove a {\em pointwise upper bound} $A^Ty\le \nabla f(\bar{x})$ where $\bar{x}$ is the final primal solution. This allows them to lower bound the dual objective by $\sum_{k=1}^m y_k$ because  $f^*(\nabla f(\bar{x}))\le 0$ for any $\bar{x}$ (see Lemma 4(d) in~\cite{ABCCCG0KNNP16}). Moreover, proving the pointwise upper bound $A^Ty\le \nabla f(\bar{x})$ is similar to the task of showing dual feasibility in the {\em linear} case~\cite{BN09,GN14} where $\nabla f(\bar{x})$ corresponds to the (fixed) primal cost coefficients.

Below we give a simple example with an $\ell_q$-norm objective where the pointwise upper bound $A^Ty\le \nabla f(\bar{x})$ is  not satisfied by the online primal-dual algorithm unless the dual solution $y$ is scaled down by a large (i.e. polynomial) factor. This means that one cannot obtain a sub-polynomial competitive ratio for \eqref{eq:covering-prob} using this approach directly.  

 Consider an instance with objective function $f(x)=\|x\|_2=\sqrt{\sum_{i=1}^n x_i^2}$.  There are $m=\sqrt{n}$ covering constraints, where the  $k^{th}$ constraint is $\sum_{i=k(m-1)+1}^{km} x_i \ge 1$. Note that each variable appears in only one constraint. Let $P$ be the value of primal objective and $D$ be the value of dual objective at any time. Suppose that the rate of increase of the primal objective is at most $\alpha$ times that of the dual; $\alpha$ corresponds to the competitive ratio in the online primal-dual algorithm. 
Upon arrival of any constraint $k$, it follows from the primal updates that all the variables $\{x_i\}_{i=k(m-1)+1}^{km}$ increase from $0$ to $\frac{1}{m}$. So the  increase in $P$ due to constraint $k$ is $(\sqrt{k}-\sqrt{k-1})\frac{1}{\sqrt{m}}$ for iteration $k$. This means that the increase in $D$ is at least $\frac{1}{\alpha}(\sqrt{k}-\sqrt{k-1})\frac{1}{\sqrt{m}}$, and so  $y_k\ge \frac{1}{\alpha}(\sqrt{k}-\sqrt{k-1})\frac{1}{\sqrt{m}}$. Finally, since $\bar{x}=\frac{1}{m}\mathbf{1}$, we know that $\nabla f(\bar{x})=\frac{1}{m}\mathbf{1}$ (recall $n=m^2$). On the other hand, $(A^Ty)_1 = y_1\ge \frac{1}{\alpha\sqrt{m}}$. Therefore, in order to guarantee $A^Ty\le \nabla f(\bar{x})$ we must have $\alpha\ge \sqrt{m} = n^{1/4}$.

\paragraph*{Our approach} First, we show that by duplicating variables and using an online separation oracle approach (as in~\cite{AAABN06}) one can ensure that the sets $\{S_e \}_{e=1}^r$ are disjoint. This allows for a simple expression for $\nabla f$ which is useful in the later analysis. Then we utilize the specific form of the primal-dual convex programs \eqref{eq:covering-prob} and \eqref{eq:packing-prob} and an explicit expression for $\nabla f$ to show that the dual $y$ is  approximately  feasible. In particular we show that $\|y^T A(S_e)\|_{p_e} \le O(\log d\rho)\cdot c_e$ for each $e\in [r]$; here $A(S_e)$ denotes the submatrix of $A$ with columns from $S_e$. Note that this is a weaker requirement than upper bounding $A^Ty$ pointwise by $\nabla f(\bar{x})$. 
 
 In order to  bound $\|y^T A(S_e)\|_{p_e}$, we analyze  each $e\in [r]$ separately. We partition the steps of the algorithm into {\em phases} where phase $j$ corresponds to steps where $\Phi_e = \sum_{i\in S_e} x_i^{q_e} \approx \theta^j$; here $\theta>1$ is a parameter that depends on $q_e$. The number of phases can be bounded using the fact that $\Phi_e$ is monotonically increasing. By triangle inequality we  upper bound $\|y^T A(S_e)\|_{p_e}$ by 
 $\sum_j \|y_{(j)}^T A(S_e)\|_{p_e}$ where $y_{(j)}$ denotes the dual variables that arrive in phase $j$. And in each phase $j$,  we  can upper bound  $\|y_{(j)}^T A(S_e)\|_{p_e}$ using the differential equations for the primal and dual updates.  
 
\paragraph*{Applications} We also provide two applications of Theorem~\ref{thm:main}. 

{\bf Non-uniform multicommodity buy-at-bulk.} This is a well-studied network design problem in the offline setting~\cite{CK05,CHKS10}. For its online version, the first poly-logarithmic competitive ratio was obtained recently in~\cite{ECKP15}. A key step in this result was a fractional online algorithm for a certain mixed packing-covering LP. We improve the competitive ratio of this step from $O(\log^3n)$ to $O(\log n)$ which leads to a corresponding improvement in the final result of \cite{ECKP15}. See Theorem~\ref{thm:bab}. 

{\bf Throughput maximization with $\ell_p$-norm capacities.} The online problem of maximizing throughput subject to edge capacities is well studied and a tight logarithmic competitive ratio is known~\cite{AAP93,BN09}.  We consider the generalization where instead of individual edge capacities, we can have capacity constraints on subsets as follows. A $\ell_p$-norm capacity of $c$ for some subset $S$ of edges means that the $\ell_{p}$-norm of the loads on edges of $S$ must be at most $c$. We show  that one can obtain a randomized log-competitive algorithm even in this setting, which generalizes the case with edge-capacities. See Theorem~\ref{thm:routing}.

\subsection{Related Work}
The online primal-dual framework for linear programs~\cite{BN-mono} is fairly well understood. Tight results are known for 
the class of packing and covering LPs~\cite{BN09,GN14},  with competitive ratio $O(\log d)$ for covering LPs and $O(\log d\rho)$ for packing LPs; here $d$ is the row-sparsity and $\rho$ is the ratio of the maximum to minimum entries in  the constraint matrix. Such LPs are very useful because they correspond to the LP relaxations of many combinatorial optimization problems. Combining the online LP solver with suitable online rounding schemes, good online algorithms have been obtained for many problems,  eg. set cover~\cite{AAABN03}, group Steiner tree~\cite{AAABN06}, caching~\cite{BBN12} and ad-auctions~\cite{BJN07}. Online algorithms for LPs with mixed packing and covering constraints were obtained in \cite{ABFP13}; the competitive ratio was improved in~\cite{ABCCCG0KNNP16}. Such mixed packing/covering LPs were also used to obtain an online algorithm for capacitated facility location~\cite{ABFP13}. A more complex mixed packing/covering LP was used recently in~\cite{ECKP15} to obtain online algorithms for non-uniform buy-at-bulk network design: as an application of our result, we obtain a simpler and better (by two log-factors) online algorithm for this problem. 

There have also been a number of results utilizing the online primal-dual framework with {\em convex} objectives for specific problems, eg. matching~\cite{DJ12}, caching~\cite{MS15}, energy-efficient scheduling~\cite{DH14,GKP11} and welfare maximization~\cite{BGMS11,HK15}. All of these results involve separable convex/concave functions. Recently, \cite{ABCCCG0KNNP16} considered packing/covering problems with general (non-separable) convex objectives,  but (as discussed previously) this result requires a monotone gradient assumption on the convex function. The sum of $\ell_q$-norm objectives considered in this paper  does not satisfy this condition. While our primal-dual algorithm is identical to~\cite{ABCCCG0KNNP16}, we need new techniques in the analysis.

All the results above (as well as ours) involve convex objectives and linear constraints. We note that~\cite{EKN16} obtained online primal-dual algorithms for certain semidefinite programs (i.e. involving non-linear constraints). While both our result and \cite{EKN16} generalize packing/covering LPs, they are not directly comparable.  

We also note that online algorithms with $\ell_q$-norm objectives have been studied previously for many scheduling problems, eg.~\cite{AAGKKV95,BP10}. These results use different approaches and are not directly comparable to ours. More recently~\cite{ACP14} used ideas from the online primal-dual approach in an online algorithm for unrelated machine scheduling with $\ell_p$-norm objectives as well as startup costs. However, the algorithm in~\cite{ACP14}   was tailored to their scheduling setting and we do not currently see a connection between our result and~\cite{ACP14}.

\section{Preliminaries}\label{sec:prelim}
Recall the primal covering problem~\eqref{eq:covering-prob} and its dual packing problem~\eqref{eq:packing-prob}. In the online setting, the constraints in the primal  and variables in the dual arrive over time. We need to maintain  monotonically increasing primal ($x$) and dual ($y$) solutions.    We first derive the dual program~\eqref{eq:packing-prob} from~\eqref{eq:covering-prob}  using Lagrangian duality, and then  show that one can assume that the sets $\{S_e\}_{e=1}^r$ are {\em disjoint} without loss of generality. The disjointness assumption leads to a much simpler expression for $\nabla f$ that will be used in Section~\ref{sec:algo}.  

\paragraph{Deriving the dual problem.} 
Let $f_e(x)=c_e\|x(S_e)\|_{q_e}$ and $f(x)=\sum_{e=1}^r f_e(x)$.  
The Lagrangian dual of problem~\eqref{eq:covering-prob} is given by:
\begin{align}
&\sup_{y\ge 0}\inf_{x\ge 0} \,\sum_{e=1}^r c_e\|x(S_e)\|_{q_e}+y^T(\mathbf{1}-Ax)\\
=&\sup_{y\ge 0}\sum_{k=1}^my_k-\sup_{x\ge 0} \,\left((A^Ty)^Tx-\sum_{e=1}^r c_e\|x(S_e)\|_{q_e}\right)\\
=&\sup_{y\ge 0}\sum_{k=1}^my_k-f^*(A^Ty)\label{apx:eq_L}
\end{align}
where $f^*(\cdot)$ is the conjugate function of $f(\cdot)$. Let $\mu=A^Ty$. $\mu\ge 0$ since $y\ge0$ and $A$ is a nonnegative matrix.  From~\cite{bertsekas2009convex}, since $f(x)$ is closed, proper, continuous and convex, we have
\begin{align*}
f^*(\mu)=f_1^*(\mu)\oplus\cdots\oplus f_r^*(\mu)=\inf_{\mu_1+\cdots+\mu_r=\mu}\left\{\sum_{e=1}^rf_e^*(\mu_e)\right\},
\end{align*}
where $\oplus$ is the infimal convolution.

Since $f_e(x)=c_e\|x(S_e)\|_{q_e}$, $f^*_e(\mu_e)=\sup_{x\ge 0} \,\mu_e(S_e)^Tx(S_e)-c_e\|x(S_e)\|_{q_e}$. Let $\|\cdot\|_{p_e}$ be the dual norm of $\|\cdot\|_{q_e}$. By the definition of dual norm, if $\|\mu_e(S_e)\|_{p_e}>c_e$, there exists $z\in\R^{|S_e|}$ with $\|z\|_{q_e}\le 1$ such that $\mu_e(S_e)^Tz>c_e$. Since $\mu_e\ge 0$, we can further require $z\ge 0$. Then take $x(S_e)=tz$ and $t\rightarrow \infty$, we have
$$\mu_e(S_e)^Tx(S_e)-c_e\|x(S_e)\|_{q_e}=t(\mu_e(S_e)^Tz-c_e\|z\|_{q_e})\rightarrow \infty.$$
On the other hand, by H{\"o}lder's inequality, $\mu_e(S_e)^Tx(S_e)\le \|\mu_e(S_e) \|_{p_e}\|x\|_{q_e}$. If $\|\mu_e(S_e) \|_{p_e}\le c_e$, for all $x$, we have $\mu_e(S_e)^Tx(S_e)- \|\mu_e(S_e) \|_{p_e}\|x\|_{q_e}\le0$. Therefore $x(S_e)=0$ is the maximizer with objective function value 0. Therefore, we have
$$f^*_e(\mu_e)=\begin{cases}
0,\quad\mbox{if } \|\mu_e(S_e)\|_{p_e}\le c_e\\
\infty,\quad\mbox{otherwise}
\end{cases}$$
Problem~\eqref{apx:eq_L} can then be reformulated as
\begin{equation*}
\max\, \left\{  \sum_{k=1}^m y_k\, :\,  A^Ty=\mu,\, \sum_{e=1}^r \mu_e = \mu,\,  \|\mu_e(S_e)\|_{p_e}\le c_e \, \forall e\in [r],\,  y\ge 0\right\}. 
\end{equation*}
which is exactly the packing problem~\eqref{eq:packing-prob}.

\paragraph{Dual problem for disjoint $S_e$.} When the sets $S_e$ are disjoint, the constraints $\sum_{e=1}^r \mu_e = \mu,\,  \|\mu_e(S_e)\|_{p_e}\le c_e$ are equivalent to $\|\mu(S_e)\|_{p_e}\le c_e$. Then the packing problem is of the form
\begin{equation}\label{eq:packing-disj}
\max\, \left\{  \sum_{k=1}^m y_k\, :\,  A^Ty=\mu,\,   \|\mu(S_e)\|_{p_e}\le c_e \, \forall e\in [r],\,  y\ge 0\right\}. 
\end{equation}
This is the dual program~\eqref{eq:packing-disj} used in Section~\ref{sec:algo}.

\paragraph{Reducing \eqref{eq:covering-prob} to the case of disjoint $S_e$.} Here, we show that any online algorithm for \eqref{eq:covering-prob} with disjoint sets $S_e$ can be converted into an online algorithm for the general case. 

\def\A{\ensuremath{{\cal A}}\xspace}
\def\ox{\bar{x}}
 \begin{lemma}\label{lem:disjoint}
	If there is a poly-time   $\alpha$-competitive algorithm for instances with disjoint $S_e$, then there is a poly-time    $O(\alpha)$-competitive algorithm for general instances.
\end{lemma}
\begin{proof}
Let \A denote an $\alpha$-competitive algorithm for disjoint $S_e$.	We assume that it is  a minimal algorithm, that is when constraint $k$ arrives it stops  increasing $x$  when $\sum_{i=1}^na_{ki}x_i=1$. (Any online algorithm can be ensured to be of this form.)

Given an instance $I$ with general $\{S_e\}_{e=1}^r$, we define an instance $J$ with disjoint $S'_e$ as follows. For each variable $x_i$, we introduce $r$ copies $x_i^{(1)},\dots,x_i^{(r)}$  where $x_i^{(e)}$ corresponds to the occurrence (if any) of $x_i$ in $S_e$. Set $S'_e$ to consist of the variables $x_i^{(e)}$ for $i\in S_e$; so $\{S'_e\}_{e=1}^r$ are disjoint. For each constraint $a^T_k x\ge 1$ in instance $I$, we introduce a family of constraints in instance $J$ which corresponds to all combinations of the $x_i^{(e)}$ variables.  
$$\sum_{i=1}^n a_{ki}\cdot x_i^{(e_i)}\ge 1,\quad \forall e_i \in [r],\, \, \forall i\in[n].$$  
Note that instances $I$ and $J$ share the same optimal offline  value by  design. Moreover, if $x$ is a feasible solution for $J$ then $\left\{\min_{e=1}^r x_i^{(e)}\right\}_{i=1}^n$ is a feasible solution for $I$. So an $\alpha$-competitive algorithm for $I$ also leads to one for $J$. However, this is not a poly-time reduction as there are exponentially many constraints in $J$. In order to deal with this, we use a separation oracle based algorithm, as in~\cite{AAABN06}.

	  \begin{algorithm}[H]
	  	\LinesNumbered
	  	When the $k^{th}$ covering constraint $\sum_{i=1}^na_{ki}x_i\ge1$ arrives in $I$\\
	  	\While {$\sum_{i=1}^n a_{ki}\cdot \min_{e=1}^r x_i^{(e)} <\frac12$}{
	  		let $e_i=\argmin_{e=1}^e x_i^{(e)}$ for all $i\in[n]$\;
	  		add constraint $\sum_{i=1}^n a_{ki}\cdot x_i^{(e_i)}\ge1$ to instance $J$ and run  algorithm \A\;
	  	}
		Output current solution $\ox_i = 2\cdot \min_{e=1}^r x_i^{(e)}$ for all $i\in[n]$.	  	
	  	\caption{Separation Oracle Based Algorithm for General $S_e$}
	  	\label{alg:sep}
	  \end{algorithm}
	  
	It is obvious that the output solution is feasible for instance $I$. As $x$ is an $\alpha$-competitive solution to $J$, the output solution is $2\alpha$-competitive for $I$. It remains  to show that Algorithm~\ref{alg:sep} runs in polynomial time upon arrival of any constraint $k$. For this, define  potential function $\psi=\sum_{i=1}^n \sum_{e=1}^r a_{ki}\cdot x_i^{(e)}$ which is monotone non-decreasing. We know that $\max_{i,e} a_{ki}x_i^{(e)}\le 1$ since  algorithm \A is  minimal. So $\psi\le rn$. In each iteration of Algorithm~\ref{alg:sep}, $\sum_{i=1}^na_{ki}\cdot x_i^{(e_i)}$ increases by at least $\frac{1}{2}$, i.e. $\psi$  also increases by at least $\frac12$. So the  number of iterations is bounded by $2rn$ which is polynomial.  
\end{proof}

Henceforth we will assume that the sets $\{S_e\}_{e=1}^r$ are  disjoint. Recall that the dual program~\eqref{eq:packing-prob} in this case reduces to~\eqref{eq:packing-disj}. 
It is easy to see that weak duality holds (Lemma~\ref{lem:weak-duality}). Strong duality also holds because~\eqref{eq:covering-prob} satisfies Slater's condition; however we do not use this fact. 

\begin{lemma}\label{lem:weak-duality}
For any pair of feasible solutions $x$ to \eqref{eq:covering-prob} and $(y,\mu)$ to \eqref{eq:packing-disj}, we have $$\sum_{e=1}^r c_e\|x(S_e)\|_{q_e} \ge \sum_{k=1}^m y_k.$$
\end{lemma}
\begin{proof}
This follows from the following inequalities:
$$\sum_{k=1}^m y_k = y^T\mathbf{1}\le y^T Ax= \mu^T x \le \sum_{e=1}^r \sum_{i\in S_e} \mu_i\cdot x_i \le \sum_{e=1}^r \|\mu(S_e)\|_{p_e} \cdot \|x(S_e)\|_{q_e} \le \sum_{e=1}^r c_e\cdot \|x(S_e)\|_{q_e}.$$ 
The first inequality is by primal feasibility; the second and last are by dual feasibility; the fourth is by H{\"o}lder's inequality.
\end{proof}
\section{Algorithm and analysis}
\label{sec:algo}  
  	Let $f(x) = \sum_{e=1}^r c_e\|x(S_e)\|_{q_e}$ denote the primal objective in~\eqref{eq:covering-prob}.
  	   
   \begin{algorithm}[H]
  	\LinesNumbered
	When the $k^{th}$ request $\sum_{i=1}^na_{ki}x_i\ge1$ arrives\\
	Let $\tau$ be a continuous variable denoting the current time.\;
	\While {the constraint is unsatisfied, i.e., $\sum_{i=1}^na_{ki}x_i<1$}{
		For each $i$ with $a_{ki}>0$, increase  $x_i$ at rate 
		$\frac{\partial x_i}{\partial \tau}=\frac{a_{ki}x_i+\frac1d}{\nabla_i f(x)}=\frac{a_{ki}x_i+\frac1d}{c_ex_i^{q_e-1}}\|x(S_e)\|_{q_e}^{q_e-1}$\;
		Increase $y_k$ at rate $\frac{\partial y_k}{\partial \tau}= 1$\;
		Set $\mu = A^Ty$\;
		}
  	\caption{Algorithm for $\ell_q$-norm packing/covering\label{alg:primal-dual}}
  \end{algorithm}
 
  In order to ensure that the gradient $\nabla f$ is positive, the primal solution $x$ starts off as $\delta\cdot \mathbf{1}$ where $\delta>0$ is arbitrarily small. So we assume that the initial primal value is zero.

It is clear that the algorithm maintains	a feasible and monotonically non-decreasing primal solution $x$. The dual solution $(y,\mu)$ is also monotonically non-decreasing, but not necessarily feasible. We will show  that  $(y,\mu)$ is $O(\log \rho d)$-approximately feasible, i.e. the packing constraints in~\eqref{eq:packing-disj} are violated by at most  an $O(\log \rho d)$ factor. 

\begin{lemma}\label{lem:cost}
The primal objective $f(x)$ is at most twice the dual objective $\sum_{k=1}^m y_k$.
\end{lemma}
\begin{proof}
We will show that the rate of increase of the primal is at most twice that of the dual. 
Consider the algorithm upon the arrival of some constraint $k$. Then
 \begin{align*}
 &\frac{d f(x)}{d\tau}= \sum_{i: a_{ki}>0} \nabla_i f(x) \cdot \frac{\partial x_i}{\partial \tau} = \sum_{i: a_{ki}>0} (a_{ki}x_i+\frac1d )\le 2.
 \end{align*}
 The inequality comes from the fact that (i) the process for the $k$th constraint is terminated when $\sum_ia_{ki}x_i=1$ and (ii) the number of non-zeroes in constraint $k$ is at most $d$. Also it is clear that the dual objective increases at rate one, which finishes the proof. 
 \end{proof}

 \begin{lemma}\label{lem:constraint}
The dual solution $(y,\mu)$ is $O(\log \rho d)$-approximately feasible, i.e.  
$$\|\mu(S_e)\|_{p_e} \,\, \le \,\, O(\log \rho d)\cdot c_e, \qquad \forall e\in [r].$$
\end{lemma}
\begin{proof} Fix any $e\in [r]$. When $q_e=1$, the objective function of covering problem is reduced to the linear case $c_e\sum_{i\in S_e}x_i$ and we want to prove $\|\mu(S_e)\|_\infty\le O(\log \rho d)\cdot c_e$ for all $e\in [r]$. It is equivalent to prove that $\mu_i\le O(\log \rho d)\cdot c_e$ for all $i\in S_e$. In this case, we have
\begin{align*}
\frac{\partial x_i}{\partial \tau}&=\frac{a_{ki}\,x_i+\frac1d}{c_e},\quad \frac{\partial y_k}{\partial \tau}=1,\quad
\frac{\partial \mu_i}{\partial \tau}=a_{ki}\notag\\
\Rightarrow d\mu_i&=\frac{c_e\,a_{ki}}{a_{ki}\,x_i+\frac1d} d x_i\\
\Rightarrow \Delta\mu_i&\le\int_0^{\frac{1}{\min\{a_{ij}\}}}\frac{c_e\,a_{ki}}{a_{ki}\,x_i+\frac1d} d x_i=c_e\left(\log\left(\frac{a_{ki}}{\min\{a_{ij}\}}+\frac{1}{d}\right)+\log(d)\right)=O(\log \rho d)\cdot c_e
\end{align*}	
	
The case $q_e>1$ is the main part of the analysis. In order to prove the desired upper bound on $\|\mu(S_e)\|_{p_e} $ we use a potential function $\Phi=\sum_{i\in S_e}(x_i^{q_e})$. Let  phase zero denote  the period where 
$\Phi\le \zeta: = (\frac{1}{\max\{a_{ij}\}\cdot d^2})^{q_e}$, and for each $j\ge 1$, phase $j$ is the period where $\theta^{j-1}\cdot \zeta\le \Phi<\theta^j\cdot \zeta$. Here $\theta>1$ is a parameter depending on $q_e$ that will be determined later. Note that $\Phi\le d(\frac{1}{\min\{a_{ij}\}})^{q_e}$ as variable $x_i$ will never be increased beyond $1/ \min_{j=1}^m a_{ij}$. So the number of phases is at most $3q_e\cdot \log (d \rho) / \log \theta$. Next, we bound the increase in $\|\mu(S_e)\|_{p_e} $ for each phase separately.

For any phase, we have the following equalities
  \begin{align}
  \frac{\partial x_i}{\partial \tau}&=\frac{a_{ki}x_i+\frac1d}{c_ex_i^{q_e-1}}\|x(S_e)\|_{q_e}^{q_e-1},\quad \frac{\partial y_k}{\partial \tau}=1,\quad
  \frac{\partial \mu_i}{\partial \tau}=a_{ki}\notag\\
  \Rightarrow d\mu_i&=\frac{c_e\,a_{ki}\,x_i^{q_e-1}}{(\sum_{j\in S_e}x_j^{q_e})^{1-\frac{1}{q_e}}(a_{ki}x_i+\frac1d)} d x_i\label{eq:mu_x}
  \end{align}

\smallskip \noindent {\bf Phase zero.} Suppose that each $x_i$ increases to $\alpha_i$ in phase zero. From~\eqref{eq:mu_x} we have
$$d\mu_i\le\frac{d\,c_e\,a_{ki}\,x_i^{q_e-1}}{(\sum_{j\in S_e}x_j^{q_e})^{1-\frac{1}{q_e}}} d x_i \qquad 
  \Rightarrow\qquad \frac{1}{d \, c_e \, a_{ki}}d\mu_i\le \frac{x_i^{q_e-1}}{(\sum_{j\in S_e}x_j^{q_e})^{1-\frac{1}{q_e}}} d x_i.$$
This means that the increase $\Delta \mu_i$ in $\mu_i$ can be bounded as:
$$  \frac{1}{d \,c_e \,a_{ki}}\Delta \mu_i\le \int_{\delta}^{\alpha_i}\frac{x_i^{q_e-1}}{(\sum_{j\in S_e}x_j^{q_e})^{1-\frac{1}{q_e}}} d x_i\le \int_{0}^{\alpha_i}1 d x_i\le \alpha_i.$$

  Since in phase zero,  $\Phi\le (\frac{1}{\max\{a_{ij}\}\cdot d^2})^{q_e}$, we know that each $\alpha_i\le \frac{1}{\max\{a_{ij}\}\cdot d^2}$. So $\Delta\mu_i\le \frac{c_e}{d}$ and at the end of phase zero, we have $\|\mu(S_e)\|_{p_e}\le \|\mu(S_e)\|_1 \le c_e$.  
  The last inequality is because $d\ge \max_{e} |S_e|$.
  
  \smallskip \noindent {\bf Phase $j\ge 1$.} Let $\Phi_0$ and $\Phi_1$ be the value of $\Phi$ at the beginning and end of this phase respectively. In phase $j$, suppose that each $x_i$ increases from $s_i$ to $t_i$. Then,  
  \begin{align*}
d\mu_i&=\frac{c_e\, a_{ki}\, x_i^{q_e-1}}{(\sum_{j\in S_e}x_j^{q_e})^{1-\frac{1}{q_e}}(a_{ki}x_i+\frac1d)} d x_i\le \frac{c_e x_i^{q_e-2}}{(\sum_{j\in S_e}x_j^{q_e})^{1-\frac{1}{q_e}}} d x_i
\end{align*}
So the increase $\Delta \mu_i$ in $\mu_i$ during this phase is:
$$ \Delta\mu_i \le\int_{s_i}^{t_i}\frac{c_ex_i^{q_e-2}}{(\sum_{j\in S_e}x_j^{q_e})^{1-\frac{1}{q_e}}}dx_i.$$
     Note that variables $x_{i'}$ for $i'\ne i$ can also increase in this phase: so we cannot directly bound the above  integral. This is precisely where the potential $\Phi$ is useful. We know that throughout this phase, $\sum_{i\in S_e} x_i^{q_e}\ge\Phi_0$. So,
$$  \Delta\mu_i  \le c_e\int_{s_i}^{t_i}\frac{x_i^{q_e-2}}{\Phi_0^{1-\frac{1}{q_e}}}dx_i=c_e\frac{t_i^{q_e-1}-s_i^{q_e-1}}{(q_e-1)\Phi_0^{1-\frac{1}{q_e}}} = c_e\frac{t_i^{q_e-1}-s_i^{q_e-1}}{(q_e-1)\Phi_0^{ \frac{1}{p_e}}}.$$
Above we used the assumption that $q_e>1$ in evaluating the integral. Now,
  \begin{align*}
  (\Delta\mu_i)^{p_e}&\le \frac{c_e^{p_e}}{(q_e-1)^{p_e} \, \Phi_0} \cdot \left(t_i^{q_e-1}-s_i^{q_e-1}\right)^{p_e}\le \frac{c_e^{p_e}}{(q_e-1)^{p_e} \, \Phi_0} \cdot \left(t_i^{(q_e-1)p_e}-s_i^{(q_e-1)p_e}\right) \\
& = \frac{c_e^{p_e}}{(q_e-1)^{p_e} \, \Phi_0} \cdot \left(t_i^{q_e }-s_i^{q_e}\right) 
  \end{align*}
The first inequality above uses the fact that $(z_1+z_2)^{p_e} \ge z_1^{p_e} + z_2^{p_e}$ for any $p_e\ge 1$ and $z_1,z_2\ge 0$, with $z_1=s_i^{q_e-1}$ and $z_2=t_i^{q_e-1} -s_i^{q_e-1}$. The last equality uses $\frac{1}{p_e}+\frac{1}{q_e}=1$.

We can now bound
$$\sum_{i\in S_e}  (\Delta\mu_i)^{p_e} \le \frac{c_e^{p_e}}{(q_e-1)^{p_e} \, \Phi_0} \cdot \sum_{i\in S_e}  \left(t_i^{q_e }-s_i^{q_e}\right) = \frac{c_e^{p_e}}{(q_e-1)^{p_e} \, \Phi_0} \left( \Phi_1-\Phi_0\right) \le \frac{c_e^{p_e}}{(q_e-1)^{p_e}} (\theta-1)$$

Let $\mu_j\in \R^{|S_e|}$ denote the increase in $\mu(S_e)$ during phase $j$. It follows from the above that  $\|\mu_j\|_{p_e}\le \frac{c_e}{q_e-1}(\theta-1)^{1/p_e}$. 

\smallskip
\noindent {\bf Combining across phases.} Note that $\mu=\sum_{j\ge 0}\mu_j$. By triangle inequality, we have 
\begin{equation}\label{eq:dual-apx}
\|\mu\|_{p_e}\le \sum_{j\ge 0} \|\mu_j\|_{p_e} \le c_e+ \sum_{j\ge 1} \|\mu_j\|_{p_e}  \le c_e \left( 1 + \frac{3q_e(\theta-1)^{1/p_e}}{(q_e-1)\log \theta}\cdot \log(d\rho)\right)
\end{equation}
To complete the proof we show next that for any $q_e>1$, there is some choice of $\theta>1$ such that the right-hand-side above is $O(\log (d\rho))\cdot c_e$. 
\begin{itemize}
\item If $q_e\ge 2$ then setting $\theta=2$, we have $\frac{3q_e }{(q_e-1)}(\theta-1)^{1/p_e}/\log \theta \le 6$.
\item If $1<q_e<2$ then set $\theta=1+(q_e-1)^{-\epsilon p_e}$, where $\epsilon =\frac{1}{-\log(q_e-1)}>0$. We have
  \begin{align*}
  \frac{(\theta-1)^{\frac1{p_e}}}{\log\theta}\le\frac{(\theta-1)^{\frac1{p_e}}}{\log(q_e-1)^{-\epsilon p_e}}
  =\frac{(q_e-1)^{-\epsilon}}{\log(q_e-1)^{-\epsilon p_e}}=\frac{(q_e-1)^{-\epsilon}}{-\epsilon p_e\log(q_e-1)}
=\frac{(q_e-1)^{-\epsilon} }{p_e} = \frac{2}{p_e}.
   \end{align*}
   The first inequality above uses that $\theta-1= (q_e-1)^{-\epsilon p_e} >1$. Thus we have 
$$\frac{3q_e (\theta-1)^{1/p_e}}{(q_e-1)\log \theta }\le \frac{6q_e}{(q_e-1)p_e}=6,$$
where the last equality uses $\frac{1}{p_e}+\frac{1}{q_e}=1$.
\end{itemize}

So in either case we have that the right-hand-side of~\eqref{eq:dual-apx} is at most $(1+6\log(d\rho))\cdot c_e$.   \end{proof}

Combining Lemmas~\ref{lem:weak-duality}, \ref{lem:cost} and \ref{lem:constraint}, we obtain Theorem~\ref{thm:main}. 

  \section{Applications}
  \subsection{Online Buy-at-Bulk Network Design}
In the non-uniform buy-at-bulk problem, we are given a directed graph $G=(V,E)$ with a monotone subadditive cost function $g_e:\R_+\rightarrow\R_+$ on each edge $e\in E$ and a collection $\{(s_i,t_i)\}_{i=1}^m$ of $m$ source/destination pairs. The goal is to find an $s_i-t_i$ path $P_i$ for each $i\in[m]$ such that the objective $\sum_{e\in E} g_e(load_e)$ is minimized; here $load_e$ is the number of paths using $e$. An equivalent view of this problem involves two costs $c_e$ and $\ell_e$ for each edge $e\in E$ and the objective $\sum_{e\in \cup P_i} c_e + \sum_{e\in E} \ell_e\cdot load_e$.
  In the online setting, the pairs $(s_i,t_i)$ arrive over time and we need to decide on the path  $P_i$ immediately after the $i^{th}$ pair arrives. Recently, \cite{ECKP15} gave a modular online algorithm for non-uniform buy-at-bulk with competitive ratio $O(\alpha\beta\gamma\cdot \log^5n)$ where: 
  \begin{itemize}
  \item $\alpha$ is the ``junction tree'' approximation ratio,
  \item $\beta$ is the integrality gap of the natural LP for single-sink instances,
  \item $\gamma$ is the competitive ratio of an online algorithm for single-sink instances.  
  \end{itemize}
  See \cite{ECKP15} for  more details. One of the main components in this result was an $O(\log^3n)$-competitive fractional  online algorithm for a certain mixed packing/covering LP. Here we show that Theorem~\ref{thm:main} can be used to provide a better (and tight) $O(\log n)$-competitive ratio.  This leads to the following improvement:
\begin{theorem}\label{thm:bab}
There is   an $O(\alpha\beta\gamma\cdot \log^3 n)$-competitive ratio for  non-uniform buy-at-bulk, where $\alpha, \beta, \gamma$ are as above.  
\end{theorem}

\def\T{\ensuremath{{\cal T}}\xspace} 
  
\medskip  \noindent {\bf The LP relaxation.} Let $\T=\{s_i,t_i:i\in [m]\}$  denote the set of all sources/destinations. For each $i\in[m]$ and root $r\in V$  variable $z_{ir}$ denotes the extent to which both $s_i$ and $t_i$ route to/from  $r$.  For each $r\in V$ and $e\in E$, variable $x_{er}$ denotes the extent to which edge $e$ is used in the routing to root $r$. For each $r\in V$ and $u\in\T$, variables $\{f_{r,u,e}:e\in E\}$ represent a flow between $r$ and $u$.  \cite{ECKP15} relied on solving the following LP:
  \begin{align*}
  \min\quad&\sum_{r\in V}\sum_{e\in E} c_e\cdot x_{e,r}\,\,+\,\, \sum_{r\in V}\sum_{e\in E} \ell_e\cdot \sum_{u\in \T} f_{r,u,e}\\
  \mbox{s.t.}\quad& \sum_{r\in V}z_{ir} \ge 1,\qquad \forall i\in [m]\\
&   \{f_{r,s_i,e}:e\in E\} \mbox{ is a flow from $s_i$ to $r$ of $z_{ir}$ units}, \qquad \forall r\in V,\, i\in[m]\\
&   \{f_{r,t_i,e}:e\in E\} \mbox{ is a flow from $r$ to $t_i$ of $z_{ir}$ units}, \qquad \forall r\in V,\, i\in[m]\\
& f_{r,u,e}\le x_{e,r}, \qquad \forall u\in \T,\, e\in E\\
  &x, f, z \ge0
  \end{align*}
The online algorithm in \cite{ECKP15} for this LP has competitive ratio $O(D\cdot \log n)$ w.r.t. the optimal integral solution;  here $D$ is an upper bound on the length of any $s_i-t_i$ path (note that $D$ can be as large as $n$). Using a height reduction operation, they could ensure that $D=O(\log n)$ while incurring an additional $O(\log n)$-factor loss in the objective. This lead to the $O(\log^3n)$ factor for the fractional online algorithm. Here we provide an improved $O(\log n)$-competitive algorithm for this LP which does not require any bound on the path-lengths.  

\def\mc{\mathsf{MC}}

For any $r\in V$ and $u\in \T$, let $\mc(r,u)$ denote the $u-r$ (resp. $r-u$) minimum cut in the graph with edge capacities $\{f_{r,u,e} : e\in E\}$ if $u$ is a source (resp. destination).  By the max-flow min-cut theorem, it follows that $z_{ir} \le \min\left\{ \mc(r,s_i) \,,\, \mc(r,t_i) \right\}$. Using this, we can 
  combine the first three constraints of the above LP into the following:
$$\sum_{r\in V} \min\left\{ \mc(r,s_i) \,,\, \mc(r,t_i) \right\} \,\, \ge \,\, 1,\qquad \forall i\in[m].$$
For a fixed $i\in [m]$, this  constraint is  equivalent to the following. For each $r\in V$, pick either an $s_i-r$ cut (under capacities $f_{r,s_i,\star}$) or an $r-t_i$ cut (under capacities $f_{r,t_i,\star}$), and check if the total cost of these cuts is at least $1$. This leads to the  following reformulation that eliminates the $x$ and $z$ variables.
  \begin{align*}
  \min\quad&\sum_{r\in V}\sum_{e\in E} c_e\cdot \left( \max_{u\in \T} f_{r,u,e} \right)  \,\,+\,\,\sum_{r\in V}\sum_{e\in E} \ell_e\cdot \sum_{u\in \T} f_{r,u,e}\\
  \mbox{s.t.}\quad&   
  \sum_{r\in R_s} f_{r,s_i}(S_r ) \,+\, \sum_{r\in R_t} f_{r,t_i}(T_r ) \ge 1,\qquad   \forall i\in [m], \,\, \forall (R_s,R_t) \mbox{ partition of }V, \\
  &      \qquad     \forall S_r \mbox{ : $s_i-r$ cut},\, \forall r\in R_s,\,\,  \forall  T_r  \mbox{ : $r-t_i$ cut},\, \forall r\in R_t\\
  & f \ge 0.
  \end{align*}

\ignore{
  \begin{align*}
  \min\quad&\sum_{r\in V}\sum_{e\in E} c_e\cdot \left( \max_{u\in \T} f_{r,u,e} \right)  \,\,+\,\,\sum_{r\in V}\sum_{e\in E} \ell_e\cdot \sum_{u\in \T} f_{r,u,e}\\
  \mbox{s.t.}\quad&   
  \sum_{r\in R_s} f_{r,s_i}(\delta^+(S_r)) \,+\, \sum_{r\in R_t} f_{r,t_i}(\delta^-(T_r)) \ge 1,\qquad \forall (R_s,R_t) \mbox{ partition of }V, \\
  &      \qquad     \forall \{s_i\} \sse S_r  \sse V\setminus \{r\} \mbox{ for }r\in R_s,\,\,  \forall \{t_i\} \sse T_r  \sse V\setminus \{r\} \mbox{ for }r\in R_s, \,\,  
  \forall i\in [m]\\
  & f \ge 0.
  \end{align*}
}
 
  Note that $\ell_{\log(n)}$-norm is a constant approximation for $\ell_{\infty}$. Therefore we can reformulate the above  objective function (at the loss of a constant factor) as the sum of  $\ell_{\log(n)}$ and $\ell_{1}$ norms.  Our fractional solver applies to this convex covering problem, and yields an $O(\log n)$-competitive ratio (note that $\rho=1$ for this instance). In order to  get a polynomial running time, we can use the natural “separation oracle” approach (as in Section~\ref{sec:prelim}) to produce  violated covering constraints.
  
  	  \begin{algorithm}[H]
  	  	\LinesNumbered
  	  	When the $i^{th}$ request $(s_i,t_i)$ arrives\\
  	  	\Repeat{$\sum_{r\in V} \min\left\{ \mc(r,s_i) \,,\, \mc(r,t_i) \right\} \, \ge \, \frac12$}{
 			For each $r\in V$, compute $\mc(r,s_i)$ and $\mc(r,t_i)$ and the respective cuts $S_r$ and $T_r$\;
			Let $R_s = \{r\in V \,:\, \mc(r,s_i) \le \mc(r,t_i)\}$ and $R_t=V\setminus R_s$\;
			Run Algorithm~\ref{alg:primal-dual} with  constraint $ \sum_{r\in R_s} f_{r,s_i}(S_r ) \,+\, \sum_{r\in R_t} f_{r,t_i}(T_r ) \ge 1$\;
  	  	}
  	  	
  	  	\caption{Separation Oracle Based Algorithm for Buy-at-Bulk}
  	  \end{algorithm}
  	  
Each iteration above runs in polynomial time since the minimum cuts can be  computed in polynomial time. In order to bound the number of iterations, consider the potential $\psi = \sum_{e\in E} (f_{r,s_i,e} + f_{r,t_i,e} )$. Note that $0\le \psi \le 2|E|$ and each iteration increases $\psi$ by at least $\frac12$. So the number of iterations is at most $4|E|$. 

\subsection{Throughput Maximization with $\ell_p$-norm Capacities}\label{sec:routing}
   The online problem of maximizing multicommodity flow was studied in \cite{AAP93,BN09}. In this problem, we are given a directed graph with edge capacities $u(e)$. Requests $(s_i,t_i)$ arrive in an online fashion. The algorithm should choose a path between $s_i$ and $t_i$ and allocate a bandwidth of 1 on the path to serve request $i$. The total bandwidth allocated on any edge is not allowed to exceed its capacity. This is the simplest version of the multicommodity routing problem. Here we consider an extension  with $\ell_p$-norm capacity  constraints on subsets of edges. This can be used to model situations where edges are provided by multiple  agents. Each agent $j$ owns a subset $S_j$ of edges and it requires the $\ell_{p_j}$-norm of the bandwidths of these edges to be at most $c_j$. In this section we assume the $S_j$ are disjoint. Our result also applies to general $S_j$ via  a reduction to disjoint instances.
   
   In the case of overlapping $S_j$. We can reduce this dual problem to the case with disjoint $S_j$ as follows. 
   
   Define $\bar{A}$ by the following: for each column $a_{\cdot i}$ of $A$, $\bar{A}$ has $\ell$ copies of $a_{\cdot i}$ where $\ell$ is the number of subsets $S_j$ that contain $\mu_i$. Then we can define $\bar{\mu}=\bar{A}^Ty$  and $\bar{S}_j$ to consist of {\em disjoint} columns of $\bar{A}$. 
   
   Since only variables $y_k$ arrive  online and all the updates are performed with respect to $y_k$, it follows that the $\bar{\mu}$ values of all copies of column $i$ (in $A$) are the same as $\mu_i$. So the following problem is equivalent to the original problem with overlapped $S_j$:
   \begin{equation*}
   \max\, \left\{  \sum_{k=1}^m y_k\, :\,  \bar{A}^Ty=\bar{\mu},\,   \|\bar{\mu}(\bar{S}_j)\|_{p_j}\le c_j \, \forall j\in [r],\,  y\ge 0\right\}.
   \end{equation*}
   Note that the primal form of this problem corresponds to what is solved in Section~\ref{sec:algo}.

   \begin{theorem}\label{thm:routing}
   	Assume that $c_j = \Omega(\log m)\cdot |S_j|^{1/p_j}$ for each $j$. Then there is a randomized $O(\log m)$-competitive online algorithm for throughput maximization  with $\ell_p$-norm capacities, where $m$ is the number of edges in the graph.
   \end{theorem}  
   We note that a similar ``high capacity'' assumption is also needed in the linear special case~\cite{AAP93,BN09} where each $|S_j|=1$. 
   
   In a fractional version of the problem, a request can be satisfied by several paths and the allocation of bandwidth can be in range $[0,1]$ instead of being restricted from $\{0,1\}$. For request $(s_i,t_i)$, let $\mathcal{P}_i$ be the set of simple paths between $s_i$ and $t_i$. Variable $f_{i,P}$ is defined to be the amount of flow on the path $P$ for request $(s_i,t_i)$. The total profit of algorithm is the (fractional) number of requests served and the performance is measured with respect to the maximum number of requests that could be served if the requests are known beforehand. We describe the problem as a packing problem:
   \begin{align}
   	\max\quad&\sum_{i}\sum_{P\in \mathcal{P}_i} f_{i,P}\label{eq:MRD}\\
   	\mbox{s.t.}\quad& \sum_{P\in \mathcal{P}_i}f_{i,P}\le 1,&\forall i\label{eq:z}\\
   	&\sum_i\sum_{P\in \mathcal{P}_i:e\in P}f_{i,P}= \mu_e,&\forall e\label{eq:x}\\
   	&\|\mu(S_j)\|_{p_j}\le c_j,&\forall j\label{eq:lq_cons}\\
   	&f\in \R_+^n\notag
   \end{align}
   Note that single edge capacity is a special case of~\eqref{eq:lq_cons} with $|S_j|=1$ and any $p_j$.
   
   Rewrite constraint~\eqref{eq:z} as $\sum_{P\in\mathcal{P}_i}f_{i,P}=\mu_i, \, \mu_i\le 1$, we have problem~\eqref{eq:MRD} is equivalent to
   \begin{align}
   \max\quad&\sum_{i}\sum_{P\in \mathcal{P}_i} f_{i,P}\label{eq:z_c}\\
   \mbox{s.t.}\quad& \sum_{P\in \mathcal{P}_i}f_{i,P}=\mu_i,&\forall i\label{eq:x_c}\\
   &\sum_i\sum_{P\in \mathcal{P}_i:e\in P}f_{i,P}= \mu_e,&\forall e\\
   &\|\mu_i\|_1\le 1, &\forall i\\
   &\|\mu(S_j)\|_{p_j}\le c_j,&\forall j\\
   &f\in \R_+^n\notag
   \end{align}
   Let $z_i$, $x_e$ be the primal variables corresponding to constraint~\eqref{eq:z_c},~\eqref{eq:x_c} respectively. Then the primal problem is
   \begin{align}
   \min\quad&\sum_{j}c_j \|x(S_j)\|_{q_j}+\sum_i \|z_i\|_\infty\\
   \mbox{s.t.}\quad& z_i+\sum_{e\in P}x_e\ge 1,&\forall i,\ P\in\mathcal{P}_i\\
   &x,z\in \R_+^n\notag
   \end{align}
   We drop the $\|\cdot\|_\infty$ because it only contains singleton $z_i$. Hence the corresponding primal problem is  the following.
   
   \begin{align}
   	\min\quad&\sum_{j}c_j \|x(S_j)\|_{q_j}+\sum_i z_i\label{eq:MRP}\\
   	\mbox{s.t.}\quad& z_i+\sum_{e\in P}x_e\ge 1,&\forall i,\ P\in\mathcal{P}_i\\
   	&x,z\in \R_+^n\notag
   \end{align}
   where $z$ is the dual variable of constraints~\eqref{eq:z} and $x$ is the dual variable of constraints~\eqref{eq:x} to~\eqref{eq:lq_cons}. 
      In the primal form, each request is associated with exponential number of constraints. Therefore, we again need to apply a separation oracle like Algorithm~\ref{alg:sep}. This separation oracle is based on the shortest $s'_i,t_i$-path in a modified graph where we add a dummy vertex $s'_i$ and edge $(s'_i,s_i)$. Let $z_i$ be the variable corresponding to edge $(s'_i,s_i)$.
   
   \begin{algorithm}[H]
   	\LinesNumbered
   	When the $i^{th}$ request $(s_i,t_i)$ arrives\\
   	\While {shortest $s'_i,t_i$-path has value less than $\frac12$}{
   		Let $P$ be the path corresponding to the shortest $s'_i,t_i$-path\;
   		Run the Algorithm~\ref{alg:primal-dual} with request $z_i+\sum_{e\in P}x_e\ge 1$\;
   	}
   	
   	\caption{Separation Oracle Based Algorithm for Multicommodity Routing}
   \end{algorithm}
   The shortest path algorithm runs in polynomial time and it can find the constraint with $z_i+\sum_{e\in P}x_e<\frac12$. Define potential function $\psi=z_i+\sum_{e\in E}x_e$. We know that $\psi\le m$ where $m$ is the number of edges in the modified graph since our algorithm is a minimal algorithm, that is, each iteration terminates with $z_i+\sum_{e\in P}x_e= 1$. In each iteration, $\psi$ increases by at least $\frac12$. Then the total number of iteration is at most $2m$. Finally, by doubling the variables we have a feasible solution and the objective increases by a two factor.
   
   To get an integer solution, we use a simple online randomized rounding algorithm. For each request $(s_i,t_i)$, $i\in [n]$, define $X_{i,P}$ for all $P\in\mathcal{P}_i$ to be a random variable with $\Pr[X_{i,P}=1]=\frac{f_{i,P}}{8}$. Define random variable $X_{i,e}=\sum_{P\in\mathcal{P}_i:e\in P}X_{i,P}$ and $X_e=\sum_{i=1}^nX_{i,e}$. From constraint~\eqref{eq:x}, $\E(X_e)=\frac{\mu_e}{8}$. For each request $(s_i,t_i)$, the rounding algorithm will choose path $P\in\mathcal{P}_i$ with probability $\Pr[X_{i,P}=1]$ and choose no path with the remaining probability. Let $\delta=8\log m$. We apply the following Chernoff bound theorem.
   \begin{theorem}[Chernoff Bound]
   	Let $X=\sum_{i=1}^nX_i$ where $X_i=1$ with probability $p_i$ and $X_i=0$ with probability $1-p_i$, and all $X_i$ are independent. Let $\mu=\E(X)=\sum_{i=1}^np_i$. Then
   	$$\Pr[X\ge (1+\delta)\mu]\le e^{-\frac{\delta^2}{2+\delta}\mu}\mbox{ for all }\delta>0.$$
   \end{theorem}  	  
   
   By Chernoff bound, we have $$\Pr[X_e>\frac{\mu_e}{4}+\frac{\delta}{4}]\le  e^{-\frac{(1+\frac{2\delta}{\mu_e})^2}{2+(1+\frac{2\delta}{\mu_e})}\frac{\mu}{8}}\le e^{-\frac{\mu_e+2\delta}{8}}\le e^{-\frac{\delta}{4}}=\frac{1}{m^2}.$$ 
   Then by union bound,
   $$\Pr[\sum_{e\in S_j}X_e^{p_j}\le \sum_{e\in S_j}(\frac{\mu_e}{4}+\frac{\delta}{4})^{p_j}]\ge 1-\frac{1}{m}.$$ Note that
   \begin{align*}
   	\sum_{e\in S_j}X_e^{p_j}\le \sum_{e\in S_j}(\frac{\mu_e}{4}+\frac{\delta}{4})^{p_j}\le \sum_{e\in S_j}2^{p_j}(\frac{\mu_e^{p_j}}{4^{p_j}}+\frac{\delta^{p_j}}{4^{p_j}})=\frac{1}{2^{p_j}}(\|\mu(S_j)\|_{p_j}^{p_j}+|S_j|\delta^{p_j})<c_j^{p_j},
   \end{align*}
   where the last inequality is by constraint~\eqref{eq:MRP} and the assumption $c_j=\Omega(\log m)\cdot|S_j|^{\frac{1}{{p_j}}}$. Therefore, constraint~\eqref{eq:z} is satisfied with probability at least $1-\frac1m$. All the other constraints are guaranteed to be satisfied and the expected objective function value is $\frac{1}{8}\sum_{i}\sum_{P\in \mathcal{P}_i} f_{i,P}$, which is $O(\log m)$-competitive to the offline optimal. 
   \bibliographystyle{alpha}
   \bibliography{on-qnorm}

\begin{thebibliography}{10}

\bibitem{AAABN06}
Noga Alon, Baruch Awerbuch, Yossi Azar, Niv Buchbinder, and Joseph Naor.
\newblock A general approach to online network optimization problems.
\newblock {\em ACM Transactions on Algorithms}, 2(4):640--660, 2006.

\bibitem{AAABN03}
Noga Alon, Baruch Awerbuch, Yossi Azar, Niv Buchbinder, and Joseph Naor.
\newblock The online set cover problem.
\newblock {\em SIAM J. Comput.}, 39(2):361--370, 2009.

\bibitem{AAGKKV95}
Baruch Awerbuch, Yossi Azar, Edward~F. Grove, Ming{-}Yang Kao, P.~Krishnan, and
  Jeffrey~Scott Vitter.
\newblock Load balancing in the l\({}_{\mbox{p}}\) norm.
\newblock In {\em FOCS}, pages 383--391, 1995.

\bibitem{AAP93}
Baruch Awerbuch, Yossi Azar, and Serge Plotkin.
\newblock Throughput-competitive on-line routing.
\newblock In {\em FOCS}, pages 32--40. IEEE, 1993.

\bibitem{ABFP13}
Yossi Azar, Umang Bhaskar, Lisa~K. Fleischer, and Debmalya Panigrahi.
\newblock Online mixed packing and covering.
\newblock In {\em SODA}, 2013.

\bibitem{ABCCCG0KNNP16}
Yossi Azar, Niv Buchbinder, T.{-}H.~Hubert Chan, Shahar Chen, Ilan~Reuven
  Cohen, Anupam Gupta, Zhiyi Huang, Ning Kang, Viswanath Nagarajan, Joseph
  Naor, and Debmalya Panigrahi.
\newblock Online algorithms for covering and packing problems with convex
  objectives.
\newblock In {\em FOCS}, pages 148--157, 2016.

\bibitem{ACP14}
Yossi Azar, Ilan~Reuven Cohen, and Debmalya Panigrahi.
\newblock Online covering with convex objectives and applications.
\newblock {\em CoRR}, abs/1412.3507, 2014.

\bibitem{BBN12}
Nikhil Bansal, Niv Buchbinder, and Joseph Naor.
\newblock Randomized competitive algorithms for generalized caching.
\newblock {\em SIAM J. Comput.}, 41(2):391--414, 2012.

\bibitem{BP10}
Nikhil Bansal and Kirk Pruhs.
\newblock Server scheduling to balance priorities, fairness, and average
  quality of service.
\newblock {\em {SIAM} J. Comput.}, 39(7):3311--3335, 2010.

\bibitem{bertsekas2009convex}
Dimitri~P Bertsekas.
\newblock {\em Convex optimization theory}.
\newblock Athena Scientific Belmont, 2009.

\bibitem{BGMS11}
Avrim Blum, Anupam Gupta, Yishay Mansour, and Ankit Sharma.
\newblock Welfare and profit maximization with production costs.
\newblock In {\em FOCS}, pages 77--86, 2011.

\bibitem{BCGNN14}
Niv Buchbinder, Shahar Chen, Anupam Gupta, Viswanath Nagarajan, and Joseph
  Naor.
\newblock Online packing and covering framework with convex objectives.
\newblock {\em CoRR}, abs/1412.8347, 2014.

\bibitem{BJN07}
Niv Buchbinder, Kamal Jain, and Joseph Naor.
\newblock Online primal-dual algorithms for maximizing ad-auctions revenue.
\newblock In {\em {ESA}}, pages 253--264, 2007.

\bibitem{BN09}
Niv Buchbinder and Joseph Naor.
\newblock Online primal-dual algorithms for covering and packing.
\newblock {\em Math. Oper. Res.}, 34(2):270--286, 2009.

\bibitem{BN-mono}
Niv Buchbinder and Joseph~(Seffi) Naor.
\newblock The design of competitive online algorithms via a primal-dual
  approach.
\newblock {\em Found. Trends Theor. Comput. Sci.}, 3(2-3):93--263, 2007.

\bibitem{CHK15}
T.{-}H.~Hubert Chan, Zhiyi Huang, and Ning Kang.
\newblock Online convex covering and packing problems.
\newblock {\em CoRR}, abs/1502.01802, 2015.

\bibitem{CK05}
Moses Charikar and Adriana Karagiozova.
\newblock On non-uniform multicommodity buy-at-bulk network design.
\newblock In {\em STOC}, pages 176--182. ACM, 2005.

\bibitem{CHKS10}
Chandra Chekuri, Mohammad~Taghi Hajiaghayi, Guy Kortsarz, and Mohammad~R
  Salavatipour.
\newblock Approximation algorithms for nonuniform buy-at-bulk network design.
\newblock {\em SIAM J. Comput.}, 39(5):1772--1798, 2010.

\bibitem{DH14}
Nikhil~R. Devanur and Zhiyi Huang.
\newblock Primal dual gives almost optimal energy efficient online algorithms.
\newblock In {\em SODA}, pages 1123--1140, 2014.

\bibitem{DJ12}
Nikhil~R Devanur and Kamal Jain.
\newblock Online matching with concave returns.
\newblock In {\em STOC}, pages 137--144. ACM, 2012.

\bibitem{EKN16}
Noa Elad, Satyen Kale, and Joseph~(Seffi) Naor.
\newblock Online semidefinite programming.
\newblock In {\em ICALP}, pages 40:1--40:13, 2016.

\bibitem{ECKP15}
Alina Ene, Deeparnab Chakrabarty, Ravishankar Krishnaswamy, and Debmalya
  Panigrahi.
\newblock Online buy-at-bulk network design.
\newblock In {\em FOCS}, pages 545--562, 2015.

\bibitem{GKP11}
Anupam Gupta, Ravishankar Krishnaswamy, and Kirk Pruhs.
\newblock Online primal-dual for non-linear optimization with applications to
  speed scaling.
\newblock In {\em WAOA}, pages 173--186, 2012.

\bibitem{GN14}
Anupam Gupta and Viswanath Nagarajan.
\newblock Approximating sparse covering integer programs online.
\newblock {\em Math. Oper. Res.}, 39(4):998--1011, 2014.

\bibitem{HK15}
Zhiyi Huang and Anthony Kim.
\newblock Welfare maximization with production costs: {A} primal dual approach.
\newblock In {\em SODA}, pages 59--72, 2015.

\bibitem{MS15}
Ishai Menache and Mohit Singh.
\newblock Online caching with convex costs: Extended abstract.
\newblock In {\em SPAA}, pages 46--54, 2015.

\end{thebibliography}
\end{document}